\documentclass{article}
\usepackage{spconf}
\usepackage{amssymb}
\usepackage{amsthm}
\usepackage{verbatim}
\usepackage{eepic}
\usepackage{epic}
\usepackage{graphicx}
\usepackage{epsfig}

\usepackage{amsmath}
\usepackage{amsfonts}

\newcommand{\bpsi}{\mbox{\boldmath{$\Psi$}}}
\newcommand{\bphi}{\mbox{\boldmath{$\Phi$}}}

\newcommand{\bpsil}{\mbox{\boldmath{$\psi$}}}
\newcommand{\bphil}{\mbox{\boldmath{$\phi$}}}

\newcommand{\deft}{{\stackrel{\triangle}{=}}}

\newcommand{\mub}{\mu_{\operatorname{B}}}

\newcommand{\bbi}{{{\bf I}}}

\newcommand{\bbd}{{{\bf D}}}

\newcommand{\bbb}{{{\bf B}}}

\newcommand{\bbj}{{{\bf J}}}

\newcommand{\bbw}{{{\bf W}}}

\newcommand{\bbf}{{{\bf F}}}
\newcommand{\bbu}{{{\bf U}}}

\newcommand{\bbz}{{{\bf Z}}}
\newcommand{\bbm}{{{\bf M}}}
\newcommand{\bx}{{\bf x}}

\newcommand{\brv}{{\bf r}}

\newcommand{\by}{{\bf y}}
\newcommand{\ba}{{\bf a}}
\newcommand{\bb}{{\bf b}}

\newcommand{\bo}{{\bf 0}}
\newcommand{\bd}{{\bf d}}
\newcommand{\bc}{{\bf c}}

\newcommand{\I}{{\mathcal{I}}}

\newcommand{\R}{{\mathcal{R}}}

\newcommand{\CC}{{\mathbb{C}}}

\newcommand{\st}{\operatorname{s.t.} \,}

\newcommand{\bl}{\left(}
\newcommand{\br}{\right)}

\newcommand{\bba}{{\mathbf A}}

\newcommand{\obbd}{\overline{\bbd}}

\newcommand{\bv}{{\mathbf v}}

\newtheorem{theorem}{Theorem}
\newtheorem{lemma}{Lemma}
\newtheorem{proposition}{Proposition}

\title{Block-Sparsity: Coherence and Efficient Recovery}
\twoauthors{Yonina C. Eldar} {
    Technion, Haifa, Israel \\
yonina@ee.technion.ac.il}
  {Helmut B\"{o}lcskei}
{
  ETH Zurich, Zurich, Switzerland \\
boelcskei@nari.ee.ethz.ch}

\begin{document}
\ninept

 \maketitle

\begin{abstract}

We consider compressed sensing of block-sparse signals, i.e., sparse signals that
have nonzero coefficients occuring in clusters. Based on an uncertainty relation
for block-sparse signals, we define a block-coherence measure and we show that a block-version of the orthogonal
matching pursuit algorithm recovers block $k$-sparse signals in no more than $k$ steps
if the block-coherence is sufficiently small. The same condition on block-sparsity is
shown to guarantee successful recovery through a mixed $\ell_2/\ell_1$ optimization approach.
The significance of the results lies in the fact that making explicit use of block-sparsity
can yield better reconstruction properties than treating the signal
as being sparse in the conventional sense thereby ignoring the additional structure in the
problem.
\end{abstract}


\section{Introduction}

We consider compressed sensing \cite{Candes06,Donoho06} of
sparse signals that exhibit additional structure in the form of the nonzero coefficients
occuring in clusters. It is therefore natural to ask whether explicitly taking
this block sparse structure into account yields improvements over treating
the signal as a conventional sparse signal. It was shown in \cite{richb08,EM082}\/
that the answer is in the affirmative. Moreover, in \cite{richb08} the restricted amplification
property was shown to provide a sufficient condition for robust recovery of model-compressible
(which includes block-sparse) signals. It is furthermore shown in \cite{richb08} that simple modifications
of the CoSaMP algorithm \cite{Tropp08} and of iterative hard thresholding \cite{Blumensath08} yield
reconstruction algorithms for the model-based case (including block-sparsity) that exhibit provable robustness properties.
A mixed $\ell_2/\ell_1$-norm algorithm for recovering block-sparse signals was introduced in \cite{EM082}.
The block restricted isometry property defined in \cite{EM082} provides equivalence
conditions for guaranteeing recovery of block-sparse signals.

The focus of the present paper is on the notion of coherence for block-sparse signals, i.e., block-coherence,
and can be seen as extending the program laid out in \cite{Tropp04,DE03} to the block-sparse case. We
introduce a block version of the orthogonal matching pursuit algorithm (BOMP) and find a sufficient
condition on block-coherence to guarantee recovery of block $k$-sparse signals through BOMP in no more than $k$ steps.
The same condition on block-coherence is shown to guarantee successful recovery through the mixed $\ell_2/\ell_1$ optimization
approach, described in \cite{EM082, Stoj08}. These results are akin to a sufficient condition on
conventional coherence reported in \cite{Tropp04} that guarantees recovery through OMP or $\ell_{1}$-optimization. Finally, we establish an uncertainty relation for
block-sparse signals and show how the block-coherence measure defined previously occurs naturally in
this uncertainty relation.

{\bf Notation.}\/ Throughout the paper, we denote vectors in $\CC^N$ by boldface
lowercase letters, e.g., $\bx$, and matrices by boldface uppercase
letters, e.g., $\bba$. The identity matrix is written as $\bbi$ or $\bbi_d$ when the dimension is not clear
from the context. Given a matrix $\bba$, $\bba^{T}$ and $\bba^{H}$ are its transpose and conjugate transpose, respectively,
$\bba^\dagger$ is the pseudo inverse, $\R(\bba)$ denotes its range space, $\bba_{i,j}$ is the element in the $i$th row and $j$th column,
and $\ba_{\ell}$ denotes its $\ell$th column. The $\ell$th
element of a vector $\bx$ is denoted by $x_{\ell}$.  The standard
Euclidean norm is $\|\bx\|_2=\sqrt{\bx^H\bx}$,
$\|\bx\|_1=\sum_{\ell} |x_{\ell}|$ is the $\ell_1$-norm,
$\|\bx\|_{\infty}=\max_{\ell} |x_{\ell}|$ is the
$\ell_{\infty}$-norm, and $\|\bx\|_{0}$ designates the number of nonzero entries in ${\bf x}$.
The Kronecker product of the matrices
$\bba$ and $\bbb$ is written as $\bba \otimes \bbb$. The spectral
radius of $\bba$ is denoted as
$\rho(\bba)=\lambda^{1/2}_{\max}(\bba^H\bba)$, where
$\lambda_{\max}(\bbb)$ is the largest eigenvalue of the
positive-semidefinite matrix $\bbb$.

\section{Block-Sparsity}

{\bf Block-sparsity.}\/ We consider the problem of representing a vector $\by \in \CC^L$
in a given dictionary $\bbd$ of size $L \times N$ with $L<N$, so
that
\begin{equation}
\label{eq:samples} \by=\bbd\bx
\end{equation}
for a coefficient vector $\bx \in \CC^N$. We require $\bx$ to be
block-sparse, where, throughout the paper, blocks are always
assumed to be of length $d$. To define block-sparsity,
we view $\bx$ as a concatenation of blocks (of length $d$) with $\bx[\ell]$ denoting
the $\ell$th sub-block, i.e.,
\begin{equation}
\label{eq:xblock} \bx^T=[\underbrace{x_1 \,\, \ldots \,\,
x_d}_{\bx[1]} \,\,\underbrace{x_{d+1}\,\,\ldots
x_{2d}}_{\bx[2]}\,\, \ldots\,\, \underbrace{x_{N-d+1}\,\,\ldots
\,\,x_{N}}_{\bx[M]}]^T
\end{equation}
with $N=Md$. We furthermore assume that $L=Rd$ with $R$ integer. A vector $\bx \in \CC^N$ is called block $k$-sparse
if $\bx[\ell]$ has nonzero Euclidean norm for at most $k$ indices $\ell$. When
$d=1$, block-sparsity reduces to the conventional definition of sparsity as in \cite{Candes06,Donoho06}.
Denoting
\begin{equation}
\label{eq:mixed} \|\bx\|_{2,0} = \sum_{\ell=1}^M
I(\|\bx[\ell]\|_2>0)
\end{equation}
where $I(\|\bx[\ell]\|_2>0)=1$ if $\|\bx[\ell]\|_2>0$ and $0$
otherwise, a block $k$-sparse vector $\bx$ is defined as a vector that satisfies
$\|\bx\|_{2,0} \leq k$. In the remainder of the paper conventional sparsity will be referred to simply as sparsity, in contrast to block-sparsity.

{\bf Problem statement.}\/ Our goal is to provide conditions on the dictionary $\bf{D}$ ensuring that the block-sparse vector $\bx$
can be reconstructed from measurements of the form (\ref{eq:samples}) through computationally efficient algorithms.
Our approach is largely based on \cite{Tropp04,ElBr02} (and the mathematical techniques used therein) where equivalent results are
provided for the sparse case. The results in \cite{Tropp04,ElBr02}\/ are stated
in terms of the dictionary coherence. Therefore, as a first step
in our development, we extend this conventional coherence measure to block-sparsity by defining block-coherence. Before introducing the corresponding
definition, we cite
the following proposition taken from \cite{EM082}.
\begin{proposition}
\label{prop:inv} The representation (\ref{eq:samples}) is unique
if and only if $\bbd {\bf g} \neq \bo$ for every ${\bf g} \neq \bo$ that
is block $2k$-sparse.
\end{proposition}
Similarly to (\ref{eq:xblock}), we can represent $\bbd$ as a
concatenation of column-blocks $\bbd[\ell]$ of size $L \times d$:
\begin{equation}
\label{eq:dblock} \bbd=[\underbrace{\bd_1 \,\, \ldots \,\,
\bd_d}_{\bbd[1]} \,\,\underbrace{\bd_{d+1}\,\,\ldots
\bd_{2d}}_{\bbd[2]}\,\, \ldots\,\,
\underbrace{\bd_{N-d+1}\,\,\ldots \,\,\bd_{N}}_{\bbd[M]}].
\end{equation}
Since from Proposition~\ref{prop:inv} the columns of $\bbd[\ell],\,\forall \ell$,
are linearly independent, we may write
$\bbd[\ell]=\bba[\ell]\bbw_{\ell}$ where $\bba[\ell]$ consists of
orthonormal columns that span $\R(\bbd[\ell])$ and $\bbw_{\ell}$
is invertible. Denoting by $\bba$ the $L \times N$ matrix
with blocks $\bba[\ell]$, and by $\bbw$ the $N\times N$ block-diagonal
matrix with blocks $\bbw_{\ell}$, we conclude that
$\bbd=\bba\bbw$. Since $\bbw$ is block-diagonal and invertible,
$\bc=\bbw\bx$ is block-sparse with the same block-sparsity level as
$\bx$. Therefore, in the sequel, we
assume, without loss of generality, that $\bbd$ consists of orthonormal blocks, i.e.,
$\bbd^H[\ell]\bbd[\ell]=\bbi_d$. Throughout the paper, we furthermore assume that the dictionaries we consider
satisfy the condition of Proposition \ref{prop:inv}.

{\bf Block-coherence.}\/ We define the block-coherence of $\bbd$ as
\begin{equation}
\label{eq:bc} \mub=\max_{\ell, r \neq \ell}
\frac{1}{d}\rho(\bbm[\ell,r]) \quad \mbox{with}\/ \quad
\bbm[\ell,r]=\bbd^H[\ell]\bbd[r].
\end{equation}
Note that $\bbm[\ell,r]$ is the $\ell r$th $d\,\times\,d$ block of the
$N \times N$ matrix $\bbm=\bbd^H\bbd$. When $d=1$, $\mub$ reduces to the conventional
definition of coherence \cite{DH01,ElBr02,Tropp04}
\begin{equation}
\label{eq:cc} \mu=\max_{\ell, r \neq \ell}
|\bd_{\ell}^H\bd_r|.
\end{equation}
It is easy to see
that the definition in (\ref{eq:bc}) is invariant to the choice of
orthonormal basis $\bbd[\ell]$ for $\R(\bbd[\ell])$. This is because
$\rho(\bbm[\ell,r])=\rho(\bbu^{H}_{\ell} \bbm[\ell,r] \bbu_{r})$. In the remainder of the paper conventional coherence
will be referred to simply as coherence, in contrast to block-coherence.

\begin{proposition}
\label{prop:cb} The block-coherence $\mub$ satisfies $0 \leq \mub \leq 1$.
\end{proposition}
\begin{proof} Clearly $\mub \geq 0$. To prove that $\mub \leq 1$, note that
$\rho(\bba) \leq \|\bba\|$,
where $\|\bba\|$ is any matrix norm. In particular, if $\bba$ is
a $d \times d$ matrix, then
\begin{equation}
\label{eq:me} \rho(\bba) \leq \max_j \sum_{i} |\bba_{i, j}| \leq d
\max_{i,j}|\bba_{i, j}|.
\end{equation}
In our case, $\bba=\bbm[\ell,r]$. Since the columns of $\bbd$ are
normalized, all the elements of $\bbm[\ell,r]$ have absolute value
smaller than or equal to $1$, so that from (\ref{eq:me}),
$\rho(\bbm[\ell,r]) \leq d$, and hence $\mub \leq 1$.
\end{proof}

It is interesting to compare $\mub$ with the coherence $\mu$
defined in (\ref{eq:cc}) for the same dictionary $\bbd$.
\begin{proposition}
\label{prop:co} For any dictionary $\bbd$, we have $\mub \leq \mu$.
\end{proposition}
\noindent The proof follows immediately from (\ref{eq:me}).

\section{Uncertainty Relation for Block-Sparsity}

We next show how the block-coherence $\mub$ defined above naturally appears in an
uncertainty relation for block-sparse signals. This uncertainty relation
generalizes the corresponding result for the sparse case
reported in \cite{ElBr02}.

The uncertainty principle for the sparse case
is concerned with pairs of representations of a vector $\bx \in
\CC^N$ in two different orthonormal bases for $\CC^N$:
$\{\bphil_{\ell},1 \leq \ell \leq N\}$ and $\{\bpsil_{\ell},1 \leq
\ell \leq N\}$ \cite{DH01,ElBr02}. Any vector $\bx\,\in\,\CC^N$ can
be expanded uniquely in terms of each one of these bases according to:
\begin{equation}
\label{eq:xn} \bx=\sum_{\ell=1}^N a_\ell
\bphil_{\ell}=\sum_{\ell=1}^N b_\ell \bpsil_{\ell}.
\end{equation}
The uncertainty relation sets limits on the sparsity of the
decompositions (\ref{eq:xn}) for any $\bx \in \CC^N$. Specifically,
denoting $A=\|\ba\|_0$ and
$B=\|\bb\|_0$, it is shown in \cite{ElBr02} that
\begin{equation}
\label{eq:ucd} \frac{1}{2}\bl A+B \br \geq \sqrt{AB} \geq
\frac{1}{\mu(\bphi,\bpsi)}
\end{equation}
where $\mu(\bphi,\bpsi)$ is the coherence between $\bphi$ and
$\bpsi$, defined by
\begin{equation}
\label{eq:mud}
\mu(\bphi,\bpsi)=\max_{\ell,r}|\bphil^H_{\ell}\bpsil_r|.
\end{equation}

In \cite{DH01} it is shown that $1/\sqrt{N} \leq \mu(\bphi,\bpsi) \leq
1$.
We now develop an uncertainty principle for block-sparse
decompositions, analogous to (\ref{eq:ucd}). Specifically,
we find a result that is equivalent to (\ref{eq:ucd}) with $A$ and $B$ replaced
by block-sparsity levels as defined in (\ref{eq:mixed})
and $\mu(\bphi,\bpsi)$ replaced by the block-coherence between the
orthonormal bases considered, as defined in (\ref{eq:bc2}).

\begin{theorem} \cite{EB08}
\label{thm:uncertainty} Let $\bphi,\bpsi$ be two unitary
matrices with $L\,\times\,d$ blocks $\{\bphi[\ell],\bpsi[\ell],1 \leq \ell \leq
M\}$ and let $\bx \in \CC^N$ satisfy
\begin{equation}
\bx=\sum_{\ell=1}^M \bphi[\ell]\ba[\ell]=\sum_{\ell=1}^M
\bpsi[\ell]\bb[\ell].
\end{equation}
Let $A=\|\ba\|_{2,0}$ and $B=\|\bb\|_{2,0}$. Then,
\begin{equation}
\label{eq:uca} \frac{1}{2}(A+B)\geq \sqrt{AB} \geq \frac{1}{d
\mub(\bphi,\bpsi)}
\end{equation}
where
\begin{equation}
\label{eq:bc2} \mub(\bphi,\bpsi)=\max_{\ell,r}
\frac{1}{d}\rho(\bphi^H[\ell]\bpsi[r]).
\end{equation}
\end{theorem}

It can easily be shown that for ${\bf D}$ consisting of the orthonormal
bases $\bphi$ and $\bpsi$, i.e., ${\bf D}=[\bphi\,\,\bpsi]$, we have $\mub(\bphi,\bpsi)=\mub$, where
$\mub$ is as defined in (\ref{eq:bc}) and associated with ${\bf D}=[\bphi\,\,\bpsi]$.


The bound provided by Theorem~\ref{thm:uncertainty} can be tighter
than that obtained by applying the conventional uncertainty
relation (\ref{eq:ucd}) to the block-sparse case.
This can be seen by using $\|\ba\|_{0} \leq  d \|\ba\|_{2,0}$,
$\|\bb\|_{0} \leq  d \|\bb\|_{2,0}$, and (\ref{eq:ucd}) to obtain
\begin{equation}
\sqrt{\|\ba\|_{2,0}\|\bb\|_{2,0}} \geq \frac{1}{d
\mu}.
\end{equation}
Since $\mub\,\le\,\mu$, this bound can be looser than (\ref{eq:uca}).

\subsection{Block-incoherent dictionaries}

As already noted, in the sparse case (i.e., $d=1$) for any two orthonormal bases ${\bf \Phi}$ and ${\bf \Psi}$, we have $\mu\,\ge\,1/\sqrt{N}$. We next
show that the block-coherence satisfies a similar inequality, namely $\mub\,\ge\,1/\sqrt{dN}$. Evidently, the lower bound on $\mu$ is $\sqrt{d}$ times larger
than that on $\mub$. To prove the lower bound on $\mub$, let ${\bf \Phi}$ and ${\bf \Psi}$ denote two
orthonormal bases for $\CC^N$ and let $\bba=\bphi^H\bpsi$ where $\bba[\ell,r]$ stands for
the $(\ell, r)$th $d\,\times\, d$ block of ${\bf A}$.
With $M=N/d$, we have
\begin{eqnarray}
\label{eq:bco2} M^2 \mub^2 & \geq & \sum_{\ell=1}^M \sum_{r=1}^M
\frac{1}{d^2}\lambda_{\max}(\bba^H[\ell,r]\bba[\ell,r]) \nonumber
\\ &\geq &\frac{1}{d^2}\lambda_{\max}\bl \sum_{\ell=1}^M
\sum_{r=1}^M \bba^H[\ell,r]\bba[\ell,r]\br.
\end{eqnarray}
Now, it holds that
\begin{equation}
\label{eq:bco3} \sum_{\ell=1}^M \sum_{r=1}^M
\bba^H[\ell,r]\bba[\ell,r]= \sum_{r=1}^M \bpsi^H[r]\bl
\sum_{\ell=1}^M \bphi[\ell]\bphi^H[\ell] \br \bpsi[r].
\end{equation}
\sloppy Since $\bphi$ consists of orthonormal columns,
$\sum_{\ell} \bphi[\ell]\bphi^H[\ell]=\bphi\bphi^H=\bbi_L$.
Furthermore, since $\bpsi[r]$ consists of orthonormal columns, $\forall r$, we have
$\bpsi^H[r]\bpsi[r]= \bbi_d$, $\forall r$. Therefore, (\ref{eq:bco2}) becomes
\begin{equation}
\label{eq:bco4}   \mub^2 \geq \frac{1}{M d^2}=\frac{1}{dN}
\end{equation}
which concludes the proof.

We now construct a pair of bases that achieves the lower bound on $\mub$ and
therefore has the smallest possible block-coherence. Let $\bbf$
be the DFT matrix of size $M=N/d$ with $\bbf_{\ell, r}=
(1/\sqrt{M})\exp(j 2\pi \ell r/M)$. Define $\bphi=\bbi_N$ and
\begin{equation}
\label{eq:kron} \bpsi=\bbf \otimes \bbu_d
\end{equation}
where $\bbu_d$ is an arbitrary $d\,\times\,d$ unitary matrix. For this
choice, \sloppy $\bphi^H[\ell]\bpsi[r]=\bbf_{\ell, r} \bbu_d$.
Since $\rho(\bbu_d)=1$ and $|\bbf_{\ell, r}|=1/\sqrt{M}$, we get
\begin{equation}
\label{eq:mubl} \mub=\frac{1}{d\sqrt{M}}=\frac{1}{\sqrt{dN}}.
\end{equation}
When $d=1$, this basis pair reduces to the spike-Fourier pair
which is well known to be maximally incoherent \cite{DH01}.


\section{Efficient Recovery Algorithms}

We now give operational meaning to block-coherence by showing that if it
is small enough, a block-sparse signal $\bx$ can be recovered from ${\bf y}={\bf D}\bx$ using computationally efficient
algorithms. We consider two different algorithms, namely
the mixed $\ell_2/\ell_1$ optimization
program proposed in \cite{EM082}:
\begin{eqnarray}
\label{eq:l1} \min_{\bx}  \sum_{\ell=1}^M \|\bx[\ell]\|_2 \quad
\st  \by=\bbd\bx
\end{eqnarray}
and an extension of the orthogonal matching pursuit (OMP) algorithm \cite{Mallat94} to the block-sparse case described below and termed BOMP.
We then show that both methods recover the correct block-sparse $\bx$ as long as $\mub$ associated with
$\bbd$ is small enough.

\subsection{Block OMP}

The BOMP algorithm is similar in spirit to the conventional OMP algorithm, and
can serve as a computationally attractive alternative to
(\ref{eq:l1}).

The algorithm begins by initializing the residual as $\brv_0=\by$. At the $\ell$th stage ($\ell \geq 1$)
we choose the subspace that is best matched to
$\brv_{\ell-1}$ according to:
\begin{equation}
i_{\ell}=\arg \max \|\bbd^H[i] \brv_{\ell-1}\|_2.
\end{equation}
Once the index $i_\ell$ is chosen, we find the optimal
coefficients by computing $\bx_{\ell}[i]$ as the solution to
\begin{equation}
\min \left\|\by-\sum_{i \in \I}\bbd[i]\bx_{\ell}[i]\right\|_{2}^2.
\end{equation}
Here $\I$ is the set of chosen indices $i_j,1 \leq j \leq \ell$. The
residual is then updated as
\begin{equation}
\brv_{\ell}=\by-\sum_{i \in \I}\bbd[i]\bx_{\ell}[i].
\end{equation}

\vspace*{-3mm}

\subsection{Recovery conditions}

Our main result, summarized in Theorem~\ref{thm:mu} below, is that any block $k$-sparse vector $\bx$ can be
recovered from measurements $\by=\bbd\bx$ using either the BOMP algorithm
or (\ref{eq:l1}) if the block-coherence  satisfies
$kd<(\mub^{-1}+d)/2$.  If $\bx$ was treated as a (conventional)
$kd$-sparse vector without exploiting knowledge of the block-sparse structure,
a sufficient condition for perfect recovery using OMP or (\ref{eq:l1}) for $d=1$ (a.k.a. basis pursuit)
is $kd<(\mu^{-1}+1)/2$. Since
$\mu \geq \mub$, exploiting the block structure by using BOMP or (\ref{eq:l1}) recovery is guaranteed for a potentially higher sparsity
level.

To state our results, suppose that $\bx_0$ is a length-$N$
block $k$-sparse vector, and let $\by=\bbd\bx_0$ where $\bbd$
consists of blocks $\bbd[\ell]$ with orthonormal columns.
 Let $\bbd_0$ denote the $L \times (kd)$ matrix whose blocks
correspond to the non-zero blocks of $\bx_0$, and let $\obbd_0$ be
the matrix of size $L \times (N-kd)$ which contains the columns of
$\bbd$ not in $\bbd_0$. We then have the following theorem proved in Section~\ref{sec:proofth2}.
\begin{theorem}
\label{thm:sc} Let $\bx_0\,\in\,\CC^{N}$ be a block $k$-sparse vector
with blocks of length $d$, and let $\by=\bbd\bx_0$ for a given $L
\times N$ matrix $\bbd$. A sufficient condition for the output of
the BOMP and of (\ref{eq:l1}) to equal $\bx_0$ is that
\begin{equation}
\label{eq:sc} \rho_c(\bbd_0^\dagger \obbd_0) <1
\end{equation}
where
\begin{equation}
\rho_c(\bba)=\max_{\ell} \sum_r \rho(\bba[r,\ell])
\end{equation}
and $\bba[r,\ell]$ is the $(r,\ell)$th $d\,\times\,d$ block of $\bba$.
\end{theorem}
Note that
\begin{equation}
\rho_c( \bbd_0^\dagger \obbd_0)=\max_{\ell} \rho_c(\bbd_0^\dagger
\obbd_0[\ell]).
\end{equation}
Therefore, (\ref{eq:sc}) implies that for all $\ell$,
\begin{equation}
\label{eq:scs} \rho_c(\bbd_0^\dagger \obbd_0[\ell])<1.
\end{equation}

The sufficient condition (\ref{eq:sc}) depends on ${\bf D}_{0}$ and hence on
the location of the nonzero blocks in $\bx_{0}$, which, of course, is not known in advance.
Nonetheless, as the following theorem shows, (\ref{eq:sc}) holds whenever the dictionary ${\bf D}$ has low block-coherence.
\begin{theorem} \cite{EB08}
\label{thm:mu} Let $\mub$ be the block-coherence defined by
(\ref{eq:bc}). Then (\ref{eq:sc}) is satisfied if
\begin{equation}
\label{eq:muc1} kd<\frac{1}{2} (\mub^{-1}+d).
\end{equation}
\end{theorem}
\noindent For $d=1$, we recover the results of \cite{Tropp04,DE03}.

\section{Proof of Theorem 2} \label{sec:proofth2}

We start with some definitions. For $\bx \, \in \,\CC^N$, we define the general mixed $\ell_2/\ell_p$  norm:
\begin{equation}
\label{eq:2p} \|\bx\|_{2,p} = \|\bv\|_p, \quad \mbox{where }
v_{\ell}=\|\bx[\ell]\|_2,
\end{equation}
and the $\bx[\ell]$ are consecutive length-$d$ blocks.
For an $L \times N$ matrix $\bba$ with $L=Rd$ and $N=Md$, where $R$ and $M$ are integers, we define the
mixed matrix norm (with block size $d$) as
\begin{equation}
\label{eq:2infm} \|\bba\|_{2,p} = \max_{\bx}
\frac{\|\bba\bx\|_{2,p}}{\|\bx\|_{2,p}}.
\end{equation}

The following lemma provides bounds on the mixed matrix norms for
$p=1,\infty$, which we will use in the sequel.
\begin{lemma} \cite{EB08}
\label{lemma:norms} Let $\bba$ be an $L \times N$ matrix with
$L=Rd$ and $N=Md$. Denote by $\bba[\ell,r]$ the $(\ell, r)$th $d\,\times\, d$ block of
$\bba$.
Then,
\vspace*{-2mm}
\begin{eqnarray}
\label{eq:lemi} \|\bba\|_{2,\infty} & \leq  & \max_r
\sum_{\ell} \rho(\bba[r,\ell])\, \deft \, \rho_r(\bba)  \\
\label{eq:lem1} \|\bba\|_{2,1} & \leq & \max_{\ell} \sum_r
\rho(\bba[r,\ell])\, \deft \, \rho_c(\bba).
\end{eqnarray}
\vspace*{-2mm}
In particular, $\rho_r(\bba)=\rho_c(\bba^H)$.
\end{lemma}

\subsection{Block OMP}

We begin by proving that (\ref{eq:sc}) is sufficient to ensure
recovery using the BOMP algorithm.

To prove the result, we first show that if $\brv_{\ell-1}$ is in $\R(\bbd_0$),
then the next chosen index
$i_{\ell}$ will be correct, namely it will correspond to a block
in $\bbd_0$. Assuming that this is true, it follows immediately that
$i_{1}$ is correct since clearly $\brv_0=\by$ lies in
$\R(\bbd_0)$.  Noting that $\brv_{\ell}$ lies in the space spanned
by $\by$ and $\bbd_0[i],i \in \I_{\ell}$, where $\I_{\ell}$
denotes the indices chosen up to stage $\ell$, it follows that if
$\I_{\ell}$ corresponds to correct indices, i.e., $\bbd[i]$ is a
block of $\bbd_0$ for all $i \in \I_{\ell}$, then $\brv_{\ell}$
also lies in $\R(\bbd_0$) and the next index will be correct as well.
Thus, at every step a correct subset is selected. It is also clear
that no index will be chosen twice since the new residual is
orthogonal to all the previously chosen subspaces; consequently
the correct $\bx_0$ will be recovered in $k$ steps.

It therefore remains to show that if $\brv_{\ell-1} \in \R(\bbd_0)$,
then under (\ref{eq:sc}) the next chosen index corresponds to a block
in $\bbd_0$. This is equivalent to requiring that
\begin{equation}
\label{eq:condp} z(\brv_{\ell-1})=\frac{\|\obbd_0^H
\brv_{\ell-1}\|_{2,\infty}}{\|\bbd_0^H \brv_{\ell-1}\|_{2,\infty}}<1.
\end{equation}
From the properties of the pseudo-inverse, $\R(\bbd_0)=\R(\bbd_0
\bbd_0^\dagger)$, and consequently $\bbd_0 \bbd_0^\dagger
\brv_{\ell-1}=\brv_{\ell-1}$. Since $\bbd_0\bbd_0^\dagger$ is
Hermitian,
\begin{equation}
\label{eq:dagger} (\bbd_0^\dagger)^H \bbd_0^H
\brv_{\ell-1}=\brv_{\ell-1}.
\end{equation}
Substituting (\ref{eq:dagger}) into (\ref{eq:condp}) yields $z(\brv_{\ell-1})=$
\begin{equation}
\frac{\|\obbd_0^H (\bbd_0^\dagger)^H \bbd_0^H
\brv_{\ell-1}\|_{2,\infty}}{\|\bbd_0^H \brv_{\ell-1}\|_{2,\infty}}\leq
\rho_r(\obbd_0^H (\bbd_0^\dagger)^H)=\rho_c(
\bbd_0^\dagger\obbd_0),
\end{equation}
where we used Lemma~\ref{lemma:norms}. This completes the proof.

\subsection{$\ell_2/\ell_1$ Optimization}

We now show that (\ref{eq:sc}) is also sufficient to ensure recovery
using (\ref{eq:l1}).
 To this end we rely on the
following lemma:
\begin{lemma} \cite{EB08}
\label{lemma:1inq} Suppose that $\bv$ is a length $N=Md$ vector
with $\|\bv[\ell]\|_2>0, \forall l$, and that $\bba$ is a matrix of size $L
\times N$, where $L=Rd$ and the blocks $\bba[\ell,r]$ are of size $d \times
d$. Then, $\|\bba\bv\|_{2,1} \leq \rho_c(\bba) \|\bv\|_{2,1}$. If
in addition the values of $\rho_c(\bba\bbj_{\ell})$ are not all
equal, then the inequality is strict. Here, $\bbj_{\ell}$ is an $N
\times d$ matrix that is all zero except for the $\ell$th $d\,\times\,d$ block
which equals $\bbi_d$.
\end{lemma}

To prove that (\ref{eq:l1}) recovers the correct vector $\bx_0$, let
$\bx'$ be another set of coefficients for which $\by=\bbd\bx'$.
Denote by $\bc_0$ and $\bc'$ the length $kd$ vectors consisting of
the non-zero elements of $\bx_0$ and $\bx'$, respectively. Let
$\bbd_0$ and $\bbd'$ denote the corresponding columns of $\bbd$
so that $\by=\bbd_0\bc_0=\bbd'\bc'$. From the assumption in
Proposition~\ref{prop:inv}, it follows that there cannot be two
different representations using the same blocks $\bbd_0$.
Therefore, $\bbd'$ must contain at least one block, $\bbz$, that
is not included in $\bbd_0$. From (\ref{eq:scs}),
$\rho_c(\bbd_0^\dagger \bbz)<1$. For any other block $\bbu$ in
$\bbd$, we must have that
\begin{equation}
\label{eq:nl1} \rho_c(\bbd_0^\dagger \bbu)\leq 1.
\end{equation}
 Indeed, if
$\bbu \in \bbd_0$, then $\bbu=\bbd_0[\ell]=\bbd_0\bbj_{\ell}$
where $\bbj_{\ell}$ is a matrix with $d$ columns which is all
zero, except for the $\ell$th block which is equal to $\bbi_d$. In
this case, $\bbd_0^\dagger \bbd_0[\ell]=\bbj_{\ell}$ and hence
$\rho_c(\bbd_{0}^{\dagger}\bbd_{0}[l])=\rho_c(\bbd_0^{\dagger}\bbu)=1$.
If, on the other hand, $\bbu=\obbd[\ell]$ for some $\ell$, then it follows from
(\ref{eq:scs}) that $\rho_c(\bbd_0^\dagger \bbu)<1$.

Now, suppose first that the blocks in $\bbd_0^\dagger \bbd'$ do not all have
the same spectral radius $\rho$. Then,
\begin{eqnarray}
\label{eq:l1u} \|\bc_0\|_{2,1} & = & \|\bbd_0^\dagger \bbd_0
\bc_0\|_{2,1}= \|\bbd_0^\dagger \by\|_{2,1}= \|\bbd_0^\dagger
\bbd'\bc'\|_{2,1} \nonumber \\
& < & \rho_c(\bbd_0^\dagger \bbd')\|\bc'\|_{2,1} \leq
\|\bc'\|_{2,1}
\end{eqnarray}
where the first equality stems from the fact that the columns of
$\bbd_0$ are linearly independent (a consequence of
the assumption in Proposition~\ref{prop:inv}), the first inequality follows from
Lemma~\ref{lemma:1inq} since $\|\bc'[\ell]\|_2>0$, $\forall l$, and the last
inequality follows from (\ref{eq:nl1}). If all the blocks of
$\bbd_0^\dagger \bbd'$ have identical spectral radius $\rho$, then
$\rho<1$ as for $\bbz \in \bbd'$, $\rho_c(\bbd_0^\dagger
\bbz)<1$. Repeating the calculations in (\ref{eq:l1u}), we find
that the first inequality is no longer strict. However, the second
inequality in (\ref{eq:l1u}) is strict instead so that the conclusion still holds.

Since $\|\bx_0\|_{2,1}=\|\bc_0\|_{2,1}$ and
$\|\bx'\|_{2,1}=\|\bc'\|_{2,1}$, we conclude that under
(\ref{eq:scs}), any set of coefficients used to represent the
original signal that is not equal to $\bx_0$ will result in a
larger $\ell_2/\ell_1$ norm.

\begin{footnotesize}

\end{footnotesize}

\end{document}